\documentclass[12pt]{amsart}
\usepackage[margin=3cm]{geometry}
\usepackage{url}

\newtheorem{theorem}{Theorem}
\newtheorem{lemma}[theorem]{Lemma}
\theoremstyle{definition}
\newtheorem{remark}[theorem]{Remark}
\newtheorem{example}[theorem]{Example}

\title[The freeness problem over matrix semigroups and bounded languages]{The freeness 
problem over matrix semigroups and bounded languages}
\author[\' E. Charlier]{\'Emilie Charlier}
\author[J. Honkala]{Juha Honkala}

\address[\'Emilie Charlier]{Institut de math\'ematiques, Universit\'e de Li\`ege, Belgium}
\email[\'Emilie Charlier]{echarlier@ulg.ac.be} 
\address[Juha Honkala]{Department of Mathematics and Statistics, University of Turku, Finland}
\email[Juha Honkala]{juha.honkala@utu.fi} 

\DeclareMathOperator{\val}{val}
\DeclareMathOperator{\N}{\mathbb N}
\DeclareMathOperator{\Q}{\mathbb Q}
\DeclareMathOperator{\Z}{\mathbb Z} 
\DeclareMathOperator{\matk}{\Q^{k\times k}_{\rm uptr}}
\DeclareMathOperator{\matkZ}{\Z^{k\times k}_{\rm uptr}}
\DeclareMathOperator{\mat}{\Q^{2\times 2}_{\rm uptr}} 

\begin{document}

\begin{abstract}
We study the freeness problem for matrix semigroups. We show that the freeness
problem is decidable for upper-triangular $2\times 2$ matrices with rational
entries when the products are restricted to certain bounded languages.
\end{abstract}
\maketitle

\section{Introduction}

In this paper we study the freeness problem over matrix semigroups. In
general, if $S$ is a semigroup and $X$ is a subset of $S$, we say that $X$ is
a code if for any integers $m,n\ge 1$ and any elements
$x_1,\ldots,x_m,y_1,\ldots,y_n \in X$ the equation
\[x_1x_2\ldots x_m=y_1y_2\ldots y_n\]
implies that $m=n$ and $x_i=y_i$ for $1\le i \le m$. The freeness problem over
$S$ consists of deciding whether or not a finite subset of $S$ is a code.

The freeness problem over $S$ can also be stated as follows. Suppose $\Sigma$
is a finite nonempty alphabet and $\mu:\Sigma^+\to S$ is a morphism. Then the
freeness problem over $S$ is to decide whether or not $\mu$ is injective.

For a general introduction to freeness problems over semigroups see
\cite{Cassaigne-Nicolas}.

An interesting special case of the freeness problem concerns freeness of
matrix semigroups. Let $R$ be a semiring and let $k\ge 1$ be an integer. Then
the semiring of $k\times k$ matrices (resp. upper-triangular $k\times k$
matrices) is denoted by $R^{k\times k}$ (resp. $R^{k\times k}_{\rm uptr}$). The
sets $R^{k\times k}$ and $R^{k \times k}_{\rm uptr}$ are monoids and the
freeness problem over $R^{k\times k}$ is to decide whether or not a given
morphism
\[\mu:\Sigma^*\to R^{k\times k}\]
is injective. Most cases of this problem are undecidable. In fact, Klarner,
Birget and Satterfield \cite{Klarner-Birget-Satterfield} proved that the
freeness problem over ${\N}^{3\times 3}$ is undecidable. Cassaigne, Harju and
Karhum\"aki \cite{Cassaigne-Harju-Karhumaki} improved this result by showing
that the problem remains undecidable for ${\N}^{3\times 3}_{\rm uptr}$. Both
of these undecidability results use the Post correspondence problem. Cassaigne,
Harju and Karhum\"aki also discuss the freeness problem for $2\times 2$
matrices having rational entries. This problem is still open even for
upper-triangular $2\times 2$ matrices. For some special decidable cases of the
freeness problem for $2\times 2$ matrices see \cite{Cassaigne-Harju-Karhumaki}
and \cite{Honkala:2009}.

In this paper we discuss the problem whether or not a given morphism
$\mu:\Sigma^*\to \matk$ is injective on certain bounded languages. This
approach is inspired by the well-known fact that many language theoretic
problems which are undecidable in general become decidable when restricted to
bounded languages. Our main result is that we can decide the
injectivity of a given morphism $\mu:\{x,z_1,\ldots,z_{t+1}\}^*\to \mat$ on
the language
$L_t=z_1x^*z_2x^*z_3\ldots z_tx^*z_{t+1}$ for any $t\ge 1$, provided that the
matrices $\mu(z_i)$ are nonsingular for $1\le i\le t+1$. To prove this result
we will study the representation of rational numbers in a rational base.

On the other hand, we will show that if we consider large enough matrices the
injectivity problem becomes undecidable even if restricted to certain very special
bounded languages. Hence, contrary to the common situation in language
theory, the restriction of the freeness problem over bounded languages remains
undecidable. The proof of our undecidability result will use a reduction to
Hilbert's tenth problem in a way which is commonly used to obtain various
undecidability results for rational power seris (see \cite{Kuich-Salomaa}) and
which is also used in \cite{Bell-Halava-Harju-Karhumaki-Potapov} to study the mortality problem for
products of matrices.

\section{Results and examples}

As usual, $\Z$ and $\Q$ are the sets of integers and rational numbers. If
$k\ge 1$ is an integer, the
set of $k\times k$ matrices having integer (resp. rational) entries is denoted
by $\Z^{k\times k}$ (resp. $\Q^{k\times k}$) and the set of upper-triangular
$k\times k$ matrices is denoted by $\matkZ$ (resp. $\matk$).

We will consider two special families of bounded languages. Suppose $t\ge 1$
is a positive integer. Let
\[\Sigma_t=\{x,z_1,\ldots,z_{t+1}\}\]
be an alphabet having $t+2$ different letters and let
\[\Delta=\{x,y,z_1,z_2\}\]
be an alphabet having four different letters. Define the languages
$L_t\subseteq \Sigma_t^*$ and $K_t\subseteq \Delta^*$ by
\[ L_t=z_1x^*z_2x^*z_3\cdots z_tx^*z_{t+1}\]
and
\[K_t=z_1(x^*y)^{t-1}x^*z_2.\]

We can now state our results.

\begin{theorem}\label{the:main1}
Let $t$ be a positive integer. It is decidable whether or not a given morphism
\[\mu\colon \Sigma_t^* \to \mat\]
such that $\mu(z_i)$
is nonsingular for $i=1,\ldots,t+1$, is injective on $L_t$.
\end{theorem}

\begin{theorem}\label{the:main2} 
There exist two positive integers $k$ and $t$ such that there is no algorithm
to decide whether a given morphism
\[\mu\colon \Delta^* \to \Z^{k\times k}_{\rm uptr}\]
is injective on $K_t$.
\end{theorem}

We will continue with examples which illustrate the problem considered in
Theorem~\ref{the:main1}. In the examples we assume that $t$ is a positive
integer,
\[\mu:\Sigma_t^*\to \mat\]
is a morphism such that $\mu(z_i)$ is nonsingular for $i=1,\ldots,t+1$. We denote
\[\mu(x)=M \mbox{ and } \mu(z_i)=N_i\]
for $i=1,\ldots,t+1$.

\begin{example}\label{ex:1}
Assume that $t=2$. Let $\mu(x)=\Big( \begin{array}{ll}
        3 & 0 \\
        0 & 1 \\
        \end{array} \Big)$ and let
$\mu(z_2)=\Big( \begin{array}{ll}
        2 & 1 \\
        0 & 3 \\
        \end{array} \Big)$.
Then
\[\mu(x^mz_2x^n)=\left( \begin{array}{cl}
        2\cdot 3^{m+n} & 3^m \\
        0 & 3 \\
        \end{array} \right)\] for all $m,n\in\N$. Hence $\mu$ is injective on
$L_2$.
\end{example}

\begin{example}\label{ex:2}
 Assume that $t=1$. Let $M = c
        \Big( \begin{array}{cc}
        1 & b \\
        0 & 1 \\
        \end{array} \Big)
$ where $b,c\in\Q$ and $c\neq 0$. Then
\[M^{^n} = c^n
        \Big( \begin{array}{cc}
        1 & nb \\
        0 & 1 \\
        \end{array} \Big)\]
for all $n\geq 0$. It follows that there exist different integers $m,n\ge 0$
such that
\[M^{^m}=M^{^n}\]
if and only if $c\in\{-1,1\}$ and $b=0$. Hence $\mu$ is injective on $L_1$ if
and only if $c \not\in \{-1,1\}$ or $b\ne 0$.
\end{example}

\begin{example}\label{ex:3}
 Assume that $t=2$ and let $M$ be as in Example~\ref{ex:2}. Let
\[N_2 = \Big( \begin{array}{cc}
        A_2 & B_2 \\
        0 & C_2 \\
        \end{array} \Big)\]
where $A_2,B_2,C_2\in \Q$. Then
\[M^{^m}N_2M^{^n} =c^{m+n} \Big( \begin{array}{cc}
        A_2 & A_2bn+B_2+C_2bm \\
        0 & C_2 \\
        \end{array} \Big)\]
for all $m,n\ge 0$. This implies that if $c\not\in \{-1,1\}$, then $\mu$ is
injective if and only if $A_2b\ne C_2b$. If $c\in \{-1,1\}$, then $\mu$ is not
injective on $L_2$.
\end{example}

\begin{example}\label{ex:4}
Assume that $t\ge 3$. Let $M$ and $N_2$ be as in Example~\ref{ex:3} and let
\[N_3 = \Big( \begin{array}{cc}
        A_3 & B_3 \\
        0 & C_3 \\
        \end{array} \Big)\]
where $A_3,B_3,C_3\in \Q$. Then we can find two different triples
$(m_1,m_2,m_3)$ and $(n_1,n_2,n_3)$ of nonnegative integers such that
\[m_1+m_2+m_3=n_1+n_2+n_3\]
and
\[C_2C_3m_1+A_2C_3m_2+A_2A_3m_3=C_2C_3n_1+A_2C_3n_2+A_2A_3n_3.\]
This implies that
\[M^{^{m_1}}N_2M^{^{m_2}}N_3M^{^{m_3}}=M^{^{n_1}}N_2M^{^{n_2}}N_3M^{^{n_3}}\]
which shows that $\mu$ is not injective on $L_t$.
\end{example}

\section{Proof of Theorem~\ref{the:main1}}
\subsection{From matrices to representations of rational numbers}

For any $r\in\Q\setminus\{0\}$ and any word $w=w_{n-1}\cdots w_1w_0$ (where the $w_i$'s are any digits),
we define the {\em value of $w$ with respect to the base $r$} to be the number
\[\val_r(w)=\sum_{i=0}^{n-1} w_i\, r^i.\]

For any number $m$, we introduce a corresponding letter denoted by $\overline{m}$
such that $\val_r(\overline{m})=m$ holds for any base $r$.

The following lemma is straightforward.

\begin{lemma}\label{lem:puissance}
Let 
$M = c
        \Big( \begin{array}{cc}
        a & b \\
        0 & 1 \\
        \end{array} \Big)
$ where $c,a,b\in\Q$. Then
\[M^n = c^n
        \left( \begin{array}{cc}
        a^n & \val_a(\overline{b}^{\, n}) \\
        0 & 1 \\
        \end{array} \right) \]
for any $n\ge1$.
\end{lemma}

The following lemma shows that in order to prove Theorem \ref{the:main1} we
can study representations of rational numbers in a rational base.

\begin{lemma}\label{lem:matrix-form}
Let  $s\ge1$ be a positive integer, let $M = c
        \Big( \begin{array}{cc}
        a & b \\
        0 & 1 \\
        \end{array} \Big)
$ with $a,b,c\in\Q$ and, for $i=1,\ldots,s+1$, let 
$N_i=\Big( \begin{array}{cc}
        A_i & B_i \\
        0 & C_i \\
        \end{array} \Big)$ with $A_i,B_i,C_i\in\Q$.
Then we can compute rational numbers $q_1,\ldots, q_{s+1},p_1,\ldots,p_s$ such that
\begin{gather}\label{eq:lem2}
N_1M^{^{m_1}}N_2
\cdots N_sM^{^{m_s}}N_{s+1} \hspace{10cm}\\ 
\hspace{2cm} = c^{m_1+\cdots +m_s}
        \left( \begin{array}{cc}
        A_1\cdots A_{s+1} a^{m_1+\cdots +m_s} & \val_a(\overline{q_1}\, \overline{p_1}^{^{m_s-1}}
        \, \overline{q_2}
        \cdots\, \overline{q_s}\, \overline{p_s}^{^{m_1-1}}\, \overline{q_{s+1}}) \\
        0 & C_1\cdots C_{s+1} \\
        \end{array} \right)\nonumber
\end{gather}
for all positive integers $m_1,\ldots,m_s$.
\end{lemma}

\begin{proof}
We proceed by induction on $s$. Suppose first that $s=1$. 
If $m_1\ge1$, Lemma~\ref{lem:puissance} implies
\begin{eqnarray*}
N_1M^{^{m_1}}N_2
&=& 
\left( \begin{array}{cc}
        A_1 & B_1 \\
        0 & C_1 \\
        \end{array} \right) c^{m_1}
\left( \begin{array}{cc}
        a^{m_1} & \val_a(\overline{b}^{\, m_1}) \\
        0 & 1 \\
        \end{array} \right) 
\left( \begin{array}{cc}
        A_2 & B_2 \\
        0 & C_2 \\
        \end{array} \right) \\
&=& c^{m_1}\left( \begin{array}{cc}
        A_1 a^{m_1} & A_1\val_a(\overline{b}^{\, m_1})+B_1 \\
        0 & C_1 \\
        \end{array} \right) 
\left( \begin{array}{cc}
        A_2 & B_2 \\
        0 & C_2 \\
        \end{array} \right)\\
&=& c^{m_1}\left( \begin{array}{cc}
        A_1A_2 a^{m_1} & A_1B_2a^{m_1}+A_1C_2\val_a(\overline{b}^{\, m_1})+B_1C_2 \\
        0 & C_1C_2 \\
        \end{array} \right) \\
&=& c^{m_1}\left( \begin{array}{cc}
        A_1A_2 a^{m_1} & \val_a\left(\overline{A_1B_2}\  \overline{A_1C_2b}^{^{m_1-1}}\, \overline{C_2(A_1b+B_1)}\right) \\
        0 & C_1C_2 \\
        \end{array} \right). 
\end{eqnarray*}
This implies the claim for $s=1$.

Let then $s\ge1$ and assume inductively that we have computed rational numbers
 $q_1,\ldots, q_{s+1},p_1,\ldots,p_s$ such that \eqref{eq:lem2} holds for all $m_1,\ldots,m_s\ge1$.
Let $m_{s+1}\ge1$ and let $N_{s+2}=\Big( \begin{array}{cc}
        A_{s+2} & B_{s+2} \\
        0 & C_{s+2} \\
        \end{array} \Big)$.
For the sake of brevity, let us denote $d_1=A_1\cdots A_{s+1}$, $d_2=C_1\cdots C_{s+1}$
and $N_{s+2}=\Big( \begin{array}{cc}
        A & B \\
        0 & C \\
        \end{array} \Big)$. Then
\[N_1M^{^{m_1}}N_2M^{^{m_2}}N_3\cdots N_{s+1}M^{^{m_{s+1}}}N_{s+2} \hspace{8cm}\]
\vspace{-0.5cm}
\begin{eqnarray*}
\hspace{0.5cm} &=& c^{m_1+\cdots +m_s}
        \left( \begin{array}{cc}
        d_1 a^{m_1+\cdots+ m_s} & T \\
        0 & d_2 \\
        \end{array} \right) c^{m_{s+1}}
\left( \begin{array}{cc}
        a^{m_{s+1}} & \val_a(\overline{b}^{\, m_{s+1}}) \\
        0 & 1 \\
        \end{array} \right) 
\left( \begin{array}{cc}
        A & B \\
        0 & C \\
        \end{array} \right)\\
&=& c^{m_1+\cdots +m_{s+1}}
        \left( \begin{array}{cc}
        d_1 a^{m_1+\cdots +m_s} & T \\
        0 & d_2 \\
        \end{array} \right)
\left( \begin{array}{cc}
        A a^{m_{s+1}} & Ba^{m_{s+1}}+C\val_a(\overline{b}^{\, m_{s+1}}) \\
        0 & C \\
        \end{array} \right) \\
&=&c^{m_1+\cdots +m_{s+1}}
        \left( \begin{array}{cc}
        d_1A a^{m_1+\cdots +m_{s+1}} & d_1a^{m_1+\cdots +m_s}(Ba^{m_{s+1}}+C\val_a(\overline{b}^{\, m_{s+1}}))+CT \\
        0 & d_2 C\\
        \end{array} \right)
\end{eqnarray*}
where $T=\val_a(\overline{q_1}\, \overline{p_1}^{^{m_s-1}}\, \overline{q_2}\, 
                \cdots\overline{q_s}\, \overline{p_s}^{^{m_1-1}}\, \overline{q_{s+1}})$.
We compute $d_1A=A_1\cdots A_{s+2}$, $d_2C=C_1\cdots C_{s+2}$ and
\begin{gather*}
d_1a^{m_1+\cdots +m_s}(Ba^{m_{s+1}}+C\val_a(\overline{b}^{\, m_{s+1}}))+CT \hspace{7.5cm}\\
= \val_a(\overline{d_1B}\  \overline{d_1Cb}^{^{m_{s+1}-1}}\, \overline{C(d_1b+q_1)}\
         \overline{Cp_1}^{^{m_s-1}}\, \overline{Cq_2}\ 
                \cdots\overline{Cq_s}\ \overline{Cp_s}^{^{m_1-1}}\, \overline{Cq_{s+1}}).
\end{gather*}
This concludes the proof.
\end{proof}

\subsection{Comparison of the representations}

If $\Sigma$ is an alphabet, we let $\hat{\Sigma}$ be the alphabet
defined by
\[\hat{\Sigma}=\left\{\left[\begin{array}{c}
        \sigma_1 \\
        \sigma_2 \\
        \end{array} \right]\colon \sigma_1, \sigma_2\in\Sigma\right\}.\]
A word in $\hat{\Sigma}^{^{ \scriptstyle *}}$ given by
\[\left[\begin{array}{c}
        \sigma_{i_1} \\
        \sigma_{j_1} \\
        \end{array} \right]
\left[\begin{array}{c}
        \sigma_{i_2} \\
        \sigma_{j_2} \\
        \end{array} \right]
\cdots
\left[\begin{array}{c}
        \sigma_{i_\ell} \\
        \sigma_{j_\ell} \\
        \end{array} \right]\]
will be written as 
\[\left[\begin{array}{c}
        \sigma_{i_1}\sigma_{i_2}\cdots \sigma_{i_\ell} \\
        \sigma_{j_1}\sigma_{j_2}\cdots \sigma_{j_\ell} \\
        \end{array} \right].\]

In what follows it is important to observe that if we have a word
$\left[\begin{array}{c}
        w_1 \\
        w_2 \\
        \end{array} \right]$
in $\hat{\Sigma}^*$ then necessarily the words $w_1$ and $w_2$ have equal
lengths.

The next lemma shows that in comparing the representations of rational numbers
we can use regular languages.

\begin{lemma}\label{lem:regular}
Let $S\subseteq \Q$ be a finite nonempty set, let $S_1=\{\overline{s}\colon s\in S\}$
and let $X=\hat{S_1}$. Let $r\in\Q\setminus\{-1,0,1\}$.
Then the language
\[ L=\left\{\left[ \begin{array}{c}
        w_1 \\
        w_2 \\
	\end{array} \right]
\in X^{^{\scriptstyle *}}\colon
\val_r(w_1)=\val_r(w_2)\right\}\]
is effectively regular.
\end{lemma}

\begin{proof}
First, observe that
\[\val_r(x_n\cdots x_1x_0)=\val_r(y_n\cdots y_1y_0)\]
holds if and only if 
\[\val_{r^{-1}}(x_0x_1\cdots x_n)=\val_{r^{-1}}(y_0y_1\cdots y_n)\]
holds (here, the $x_i$'s and $y_i$'s are digits). 
Because the class of effectively regular languages is closed under reversal, we may assume $|r|>1$ without loss of generality. 

Next, we assume without loss of generality that
\[S=\{-m+1,-m+2,\ldots,-1,0,1,\ldots,m-2,m-1\}\]
where $m$ is a positive integer.
In other words, we will assume that
\[ X=\left\{\left[\begin{array}{c}
        \overline{a} \\
        \overline{b} \\
        \end{array} \right]\colon a,b\in\{-m+1,-m+2,\ldots,-1,0,1,\ldots,m-2,m-1\}\right\}.\]

Let $r=\frac uv$, where $u,v\in\Z$ do not have any nontrivial common factor.
Let $d=\frac{2m-2}{|r|-1}$.
We define the nondeterministic automaton $\mathcal A=(Q, X,\delta,\{q_0\},\{q_0\})$ as follows:
\[Q=\{q_i\colon i\in[-d,d]\cap\Z\}\]
and
\[\delta\big(q_i,\left[ \begin{array}{c}
	\overline{a} \\
	\overline{b} \\
	\end{array} \right]\big)
=\left\{ \begin{array}{cl}
	q_j, &\ \text{ if } i+a-b=rj; \\
	\emptyset, &\ \text{ if } \frac{i+a-b}{r} \not\in[-d,d]\cap\Z.\\
	\end{array}\right.\]
We will prove $L(\mathcal A)=L^T$. (Here $L^T$ is the reversal of $L$.)

Assume first that
\[\left[ \begin{array}{c}
	\overline{a_0} \\
	\overline{b_0} \\
	\end{array} \right]
\left[ \begin{array}{c}
	\overline{a_1} \\
	\overline{b_1} \\
	\end{array} \right]
\cdots
\left[ \begin{array}{c}
	\overline{a_n} \\
	\overline{b_n} \\
	\end{array} \right]
\in L^T,\]
or, equivalently,
\begin{equation}\label{eq:H1}
a_0+a_1r+\cdots+a_nr^n=b_0+b_1r+\cdots+b_nr^n.
\end{equation}
We claim that there exist states $q_{\alpha_1},q_{\alpha_2},\ldots,q_{\alpha_{n+1}}\in Q$ such that
\begin{equation}\label{eq:TH1}
\delta\big(q_0,\left[ \begin{array}{c}
	\overline{a_0} \\
	\overline{b_0} \\
	\end{array} \right]
\left[ \begin{array}{c}
	\overline{a_1} \\
	\overline{b_1} \\
	\end{array} \right]
\cdots
\left[ \begin{array}{c}
	\overline{a_i} \\
	\overline{b_i} \\
	\end{array} \right]\big)
=q_{\alpha_{i+1}}
\end{equation}
and
\begin{equation}\label{eq:TH2}
\alpha_{i+1}+a_{i+1}+\cdots+a_nr^{n-i-1}=b_{i+1}+\cdots+b_nr^{n-i-1}
\end{equation}
hold for all $i=0,\ldots,n$.

We first show the existence of $q_{\alpha_1}$. Because \eqref{eq:H1} implies
\[a_0v^n+a_1uv^{n-1}+\cdots+a_nu^n=b_0v^n+b_1uv^{n-1}+\cdots+b_nu^n,\]
we have $a_0\equiv b_0\bmod{u}$. Hence  
\[\alpha_1=\frac{a_0-b_0}{r}=\frac{(a_0-b_0)v}{u}\]
is an integer. 
Because $|a_0|\leq m-1$ and $|b_0|\leq m-1$, we have
\[|\alpha_1|=\frac{|a_0-b_0|}{|r|}\le d,\]
and hence the state $q_{\alpha_1}$ exists.

Further, we have
\[\delta\big(q_0,\left[ \begin{array}{c}
	\overline{a_0} \\
	\overline{b_0} \\
	\end{array} \right]\big)=q_{\alpha_1}\]
and 
\[\alpha_1+a_1+a_2r+\cdots+a_nr^{n-1}=b_1+b_2r+\cdots+b_nr^{n-1}.\]
This proves the claim for $i=0$.

Assume then $j\in\{1,\ldots,n\}$ and assume that there exist $q_{\alpha_1},\ldots,q_{\alpha_j}\in Q$ such that
\eqref{eq:TH1} and \eqref{eq:TH2} hold for $i=0,\ldots,j-1$. 
From \eqref{eq:TH2} it follows 
\[\alpha_j+a_j\equiv b_j\bmod u.\]
Hence 
\[\alpha_{j+1}=\frac{\alpha_j+a_j-b_j}{r}=\frac{(\alpha_j+a_j-b_j)v}{u}\]
is an integer. Because we have
\[|\alpha_{j+1}|=\frac{|\alpha_j+a_j-b_j|}{|r|}
	\le  \frac{|\alpha_j|+|a_j-b_j|}{|r|}
	\le  \frac{d+2m-2}{|r|}=\frac{d+d(|r|-1)}{|r|}=d,\]
the state $q_{\alpha_{j+1}}$ exists. Further, we have
\[\delta\big(q_0,\left[ \begin{array}{c}
	\overline{a_0} \\
	\overline{b_0} \\
	\end{array} \right]
\left[ \begin{array}{c}
	\overline{a_1} \\
	\overline{b_1} \\
	\end{array} \right]
\cdots
\left[ \begin{array}{c}
	\overline{a_j} \\
	\overline{b_j} \\
	\end{array} \right]\big)
=\delta\big(q_{\alpha_j},
\left[ \begin{array}{c}
	\overline{a_j} \\
	\overline{b_j} \\
	\end{array} \right]\big)
=q_{\alpha_{j+1}}\]
and
\[\alpha_{j+1}+a_{j+1}+a_{j+2}r+\cdots+a_nr^{n-j-1}=b_{j+1}+b_{j+2}r+\cdots+b_nr^{n-j-1}.\]
This concludes the proof of the claim. 

From the claim it follows
\[\delta\big(q_0,\left[ \begin{array}{c}
	\overline{a_0} \\
	\overline{b_0} \\
	\end{array} \right]
\left[ \begin{array}{c}
	\overline{a_1} \\
	\overline{b_1} \\
	\end{array} \right]
\cdots
\left[ \begin{array}{c}
	\overline{a_n} \\
	\overline{b_n} \\
	\end{array} \right]\big)
=q_{\alpha_{n+1}}\]
and
\[\alpha_{n+1}=0.\]
Therefore
\[\left[ \begin{array}{c}
	\overline{a_0} \\
	\overline{b_0} \\
	\end{array} \right]
\left[ \begin{array}{c}
	\overline{a_1} \\
	\overline{b_1} \\
	\end{array} \right]
\cdots
\left[ \begin{array}{c}
	\overline{a_n} \\
	\overline{b_n} \\
	\end{array} \right]
\in L(\mathcal{A}).\]
Hence $L^T\subseteq L(\mathcal{A})$.

Suppose now that
\[\left[ \begin{array}{c}
	\overline{a_0} \\
	\overline{b_0} \\
	\end{array} \right]
\left[ \begin{array}{c}
	\overline{a_1} \\
	\overline{b_1} \\
	\end{array} \right]
\cdots
\left[ \begin{array}{c}
	\overline{a_n} \\
	\overline{b_n} \\
	\end{array} \right]
\in L(\mathcal{A}).\]
Then there exist states $q_{\alpha_0},q_{\alpha_1},\ldots,q_{\alpha_{n+1}}\in Q$ such that
\[\delta\big(q_{\alpha_i},\left[ \begin{array}{c}
	\overline{a_i} \\
	\overline{b_i} \\
	\end{array} \right]\big)
=q_{\alpha_{i+1}}\]
for $i=0,\ldots,n$ and $\alpha_0=\alpha_{n+1}=0$. 
By the definition of $\mathcal A$ we have
\[\alpha_i+a_i-b_i=r\alpha_{i+1}\]
 for $i=0,\ldots,n$. This implies
\[a_0+a_1r+\cdots+a_nr^n=b_0+b_1r+\cdots+b_nr^n.\]
Hence 
\[\left[ \begin{array}{c}
	\overline{a_0} \\
	\overline{b_0} \\
	\end{array} \right]
\left[ \begin{array}{c}
	\overline{a_1} \\
	\overline{b_1} \\
	\end{array} \right]
\cdots
\left[ \begin{array}{c}
	\overline{a_n} \\
	\overline{b_n} \\
	\end{array} \right]
\in L^T.\]
Therefore $L(\mathcal{A})\subseteq L^T$.
\end{proof}

\subsection{A decidability method for Theorem \ref{the:main1}}

We are now ready for the proof of Theorem~\ref{the:main1}.

Let $t$ be a positive integer and assume that
\[\mu\colon \Sigma_t^* \to\mat \] is a morphism such that
$\mu(z_i)$ is nonsingular for $i=1,\ldots,t+1$.

First, we consider the particular case where $\mu(x)$ is singular. 
Suppose $\mu(x)=\Big( \begin{array}{cc}
	a & b \\
	0 & 0 \\
	\end{array} \Big)$, 
the case $\mu(x)=\Big( \begin{array}{cc}
	0 & b \\
	0 & c \\
	\end{array} \Big)$ being symmetric.
Then $\mu(x^n)=a^{n-1}\mu(x)$ for all $n\ge 1$. If $t=1$, then $\mu$ in
injective on $L_1$ if and only if $a\not\in \{-1,0,1\}$. If $t\ge 2$, then
the equation $\mu(x^2z_2x)=\mu(xz_2x^2)$ implies that $\mu$ is not injective
on $L_t$. 

For the rest of the proof we suppose that $\mu(x)$ is not singular.
Let \[\mu(x)=M = c\left( \begin{array}{cc}
        a & b \\
        0 & 1 \\
        \end{array} \right)\]
and, for $i=1,\ldots,t+1$, let 
\[\mu(z_i)=N_i=\Big( \begin{array}{cc}
        A_i & B_i \\
        0 & C_i \\
        \end{array} \Big),\]
where $a,b,c,A_i,B_i,C_i \in \Q$ for $i=1,\ldots,t+1$. Because $M$ and
$N_i$ are nonsingular, $a,c,A_i,C_i$ are nonzero for $i=1,\ldots,t+1$.
 
If $a=-1$, then $M^2=c^2I$. If $t\ge 2$, then $\mu$ is not injective on $L_t$ because we
have $N_1M^2N_2=N_1N_2M^2$. If $t=1$ and $c\in \{-1,1\}$, then $\mu$ is not
injective on $L_t$ because $N_1N_2=N_1M^2N_2$. If $t=1$ and $c\not \in \{-1,1\}$, it
follows from the equation $\det(M^n)=(-c)^n$ that $\mu$ is injective on $L_t$.
 
For the rest of the proof we suppose in addition that $a\neq -1$. We suppose
also that $a\neq 1$. In fact, we have already proved Theorem~\ref{the:main1} if $a=1$ in
Examples~\ref{ex:2}, \ref{ex:3} and \ref{ex:4}.

For each subset $K\subseteq\{1,\ldots,t\}$, let 
\[L_t(K)=\{z_1x^{m_1}z_2x^{m_2}z_3\cdots z_tx^{m_t}z_{t+1}\colon m_i=0 \text{ for } i\in K,\  m_i\ge 1 \text{ for } i\not\in K\}.\]
Now $L_t$ is a disjoint union of the languages $L_t(K)$ where $K$ runs over all the subsets of $\{1,\ldots,t\}$.
Hence the morphism $\mu$ is injective on $L_t$ if and only if
\begin{itemize}
\item[(i)] \label{eq:i} for each $K\subseteq\{1,\ldots,t\}$,  $\mu$ is injective on $L_t(K)$; and
\item[(ii)]  \label{eq:ii} if $K_1,K_2\subseteq\{1,\ldots,t\}$ with $K_1\neq K_2$, 
then there does not exist two words $w_1\in L_t(K_1)$ and $w_2\in L_t(K_2)$ such that $\mu(w_1)=\mu(w_2)$.
\end{itemize}

We first prove that (ii) is decidable.
For $w_1\in L_t(K_1)$ and $w_2\in L_t(K_2)$, we have
\[\mu(w_1)=N_1'M^{^{k_1}}N_2'M^{^{k_2}}N_3'\cdots N_{s_1}'M^{^{k_{s_1}}}N_{s_1+1}'\]
and
\[\mu(w_2)=N_1''M^{^{\ell_1}}N_2''M^{^{\ell_2}}N_3''\cdots N_{s_2}''M^{^{\ell_{s_2}}}N_{s_2+1}''\]
where $s_1=t-|K_1|$, $s_2=t-|K_2|$, $k_i\ge 1$ for $i=1,\ldots,s_1$, $\ell_j\ge 1$ for $j=1,\ldots,s_2$ and
\[N_1N_2\cdots N_{t+1}=N_1'N_2'\cdots N_{s_1+1}'=N_1''N_2''\cdots N_{s_2+1}''.\]
In view of Lemma~\ref{lem:matrix-form}, deciding (ii) is equivalent to deciding the following two problems: 
\begin{enumerate}
\medskip
\item[A :] Given positive integers $s_1,s_2$ and rational numbers
$p_1,\ldots, p_{s_1}$, $q_1,\ldots,q_{s_1+1}$, $\alpha_1,\ldots,\alpha_{s_2}$, $\beta_1,\ldots,\beta_{s_2+1}$, 
decide whether there exist positive integers $k_1,\ldots,k_{s_1}$, $\ell_1,\ldots,\ell_{s_2}$
such that the two matrices
\begin{gather}\label{mat1}
 \hspace{1.35cm} c^{k_1+\cdots +k_{s_1}}
	\left( \begin{array}{cc}
        A_1\cdots A_{t+1} a^{k_1+\cdots +k_{s_1}} & \val_a(\overline{q_1}\, \overline{p_1}^{^{k_{s_1}-1}}
        \, \overline{q_2}\,
        \cdots\,\overline{q_{s_1}}\, \overline{p_{s_1}}^{^{k_1-1}}\, \overline{q_{s_1+1}}) \\
	0 & C_1\cdots C_{t+1} \\
        \end{array} \right)
\end{gather}
and
\begin{gather}\label{mat2}
\hspace{1.35cm} c^{\ell_1+\cdots +\ell_{s_2}}
	\left( \begin{array}{cc}
        A_1\cdots A_{t+1} a^{\ell_1+\cdots +\ell_{s_2}} & \val_a(\overline{\beta_1}\, \overline{\alpha_1}^{^{\ell_{s_2}-1}}\,
                 \overline{\beta_2}\,
                 \cdots\, \overline{\beta_{s_2}}\, \overline{\alpha_{s_2}}^{^{\ell_1-1}}\, \overline{\beta_{s_2+1}}) \\
	0 & C_1\cdots C_{t+1} \\
	\end{array} \right)
\end{gather}
are equal.
\medskip
\item[B :] Given a positive integer $s$ and rational numbers
$q,p_1,\ldots,p_s, q_1,\ldots,q_{s+1}$, 
decide whether there exist positive integers $k_1,\ldots,k_s$
such that the two matrices
\begin{gather}\label{mat3}
  c^{k_1+\cdots +k_s}
	\left( \begin{array}{cc}
        A_1\cdots A_{t+1} a^{k_1+\cdots +k_s} & \val_a(\overline{q_1}\, \overline{p_1}^{^{k_s-1}}
        \overline{q_2}\, \cdots\, \overline{q_s}\, \overline{p_s}^{^{k_1-1}}\, \overline{q_{s+1}}) \\
	0 & C_1\cdots C_{t+1} \\
      \end{array} \right)
\end{gather}
and
\begin{gather}\label{mat4}
	\left( \begin{array}{cc}
        A_1\cdots A_{t+1} &  q \\
	0 & C_1\cdots C_{t+1} \\
	\end{array} \right)
\end{gather}
are equal.
\end{enumerate}

\medskip
Problem B corresponds to the case where one of the subsets $K_1$ and $K_2$ is equal to $\{1,\ldots,t\}$. 
Because the products $ac$, $A_1\cdots A_{t+1}$ and $C_1\cdots C_{t+1}$ are
nonzero, a necessary condition for the equality of (\ref{mat3}) and
(\ref{mat4}) is
\begin{gather*}
a^{k_1+\cdots+ k_s}=1.
\end{gather*}
Because $a\not\in \{-1,1\}$ this condition never holds and Problem B has no
solutions.

We now turn to Problem A. Because the products $ac$, $A_1\cdots A_{t+1}$ and $C_1\cdots C_{t+1}$ are
nonzero, (\ref{mat1}) and (\ref{mat2}) are equal if and only if
\begin{gather}\label{eq:lengths1}
a^{k_1+\cdots+ k_{s_1}}= a^{\ell_1+\cdots+\ell_{s_2}},
\end{gather}

\begin{gather}\label{eq:lengths2}
c^{k_1+\cdots+ k_{s_1}}=c^{\ell_1+\cdots+\ell_{s_2}}
\end{gather}
and
\begin{gather}\label{eq:values}
\val_a(\overline{q_1}\, \overline{p_1}^{^{k_{s_1}-1}}
        \, \overline{q_2}\,
	\cdots\, \overline{q_{s_1}}\,
                \overline{p_{s_1}}^{^{k_1-1}}\, \overline{q_{s_1+1}})  
= \val_a(\overline{\beta_1}\, \overline{\alpha_1}^{^{\ell_{s_2}-1}}\,
                 \overline{\beta_2}\,
		\cdots\,   \overline{\beta_{s_2}}\,
                        \overline{\alpha_{s_2}}^{^{\ell_1-1}}\, \overline{\beta_{s_2+1}}).
\end{gather}

Because $a\not\in \{-1,0,1\}$
\eqref{eq:lengths1} and \eqref{eq:lengths2} hold if and only if
\begin{gather}\label{eq:sums}
k_1+\cdots +k_{s_1}=\ell_1+\cdots +\ell_{s_2}.
\end{gather}

Let now
$S=\{q_1,\ldots,q_{s_1+1},p_1,\ldots,p_{s_1},\beta_1,\ldots,\beta_{s_2+1},
\alpha_1,\ldots,\alpha_{s_2}\}$, let $S_1=\{\overline{s}\colon s\in S\}$  
and let $X=\hat{S_1}$. Let
\[ L=\left\{\left[ \begin{array}{c}
        u_1 \\
        u_2 \\
        \end{array} \right]
\in X^{^{\scriptstyle *}}\colon
\val_a(u_1)=\val_a(u_2)\right\}\]
and let
\begin{gather*} 
T_1=\left\{\left[ \begin{array}{c}
        u_1 \\
        u_2 \\
        \end{array} \right]
\in X^{^{\scriptstyle *}}\colon 
u_1\in \overline{q_1}\, \overline{p_1}^*\, \overline{q_2}
\cdots\, \overline{q_{s_1}}\, \overline{p_{s_1}}^*\, \overline{q_{s_1+1}} ,\,  u_2 \in \overline{\beta_1}\, \overline{\alpha_1}^*\, \overline{\beta_2}\,
\cdots\, \overline{\beta_{s_2}}\, \overline{\alpha_{s_2}}^*\, \overline{\beta_{s_2+1}} \right\}. 
\end{gather*}

By Lemma~\ref{lem:regular}, $L$ is effectively regular. 
So is clearly $T_1$. In fact, it is easy to construct a finite automaton which accepts $T_1$.
Now we can decide (ii) by checking whether or not
\[L\cap T_1=\emptyset.\]

Indeed, suppose a word $\left[ \begin{array}{c}
        u_1 \\
        u_2 \\
        \end{array} \right] \in X^*$
belongs to $L\cap T_1$. Then there exist positive integers $k_1,\ldots,k_{s_1},\ell_1,\ldots,\ell_{s_2}$ such that
\[u_1=
	\overline{q_1}\, \overline{p_1}^{^{k_{s_1}-1}}\, \overline{q_2}\,
                \cdots\overline{q_{s_1}}\,
                \overline{p_{s_1}}^{^{k_1-1}}\, \overline{q_{s_1+1}}\]
and
\[u_2=
	\overline{\beta_1}\, \overline{\alpha_1}^{^{\ell_{s_2}-1}}\,
                 \overline{\beta_2}\, \cdots
                        \overline{\beta_{s_2}}\,
                        \overline{\alpha_{s_2}}^{^{\ell_1-1}}\,
                        \overline{\beta_{s_2+1}}.\]
Because
$\left[ \begin{array}{c}
        u_1 \\
        u_2 \\
        \end{array} \right] \in L\cap T_1$, we have $\val_a(u_1)=\val_a(u_2)$ and $|u_1|=|u_2|$.
The latter condition means that
\[k_{s_1}+\cdots +k_1+1=\ell_{s_2}+\cdots+\ell_1 +1\]
which gives \eqref{eq:sums}. 
Hence \eqref{mat1} and \eqref{mat2} are equal. 
Conversely, if there exist positive integers $k_1,\ldots,k_{s_1},\ell_1,\ldots,\ell_{s_2}$ such that
the matrices \eqref{mat1} and \eqref{mat2} are equal, then
\[\left[ \begin{array}{c}
         \overline{q_1}\, \overline{p_1}^{^{k_{s_1}-1}}\, \overline{q_2}\,
                \cdots\overline{q_{s_1}}\,
                \overline{p_{s_1}}^{^{k_1-1}}\, \overline{q_{s_1+1}}\\
         \overline{\beta_1}\, \overline{\alpha_1}^{^{\ell_{s_2}-1}}\,
                 \overline{\beta_2}\, \cdots
                        \overline{\beta_{s_2}}\,
                        \overline{\alpha_{s_2}}^{^{\ell_1-1}}\,
                        \overline{\beta_{s_2+1}} \\
        \end{array} \right] \in L\cap T_1.\]

To conclude the proof of Theorem~\ref{the:main1} it remains to prove that also (i) is decidable. 
We have to decide a variant of Problem A where $s_1=s_2$,
$p_i=\alpha_i$ and $q_j=\beta_j$ for $1\le i \le s_1$, $1\le j \le s_1+1$ and
we have to find out whether there exist two different $s_1$-tuples
$(k_1,\ldots,k_{s_1})$ and $(\ell_1,\ldots,\ell_{s_1})$ of positive integers
such that (\ref{eq:values}) and (\ref{eq:sums}) hold. 
Before we can proceed as we did above in case (ii) we have to check whether there exist
different $s_1$-tuples $(k_1,\ldots,k_{s_1})$ and $(\ell_1,\ldots,\ell_{s_1})$
of positive integers such that
\[\overline{q_1}\, \overline{p_1}^{^{k_{s_1}-1}}
        \, \overline{q_2}\, \cdots\, \overline{q_{s_1}}\,
                \overline{p_{s_1}}^{^{k_1-1}}\, \overline{q_{s_1+1}}
=\overline{q_1}\, \overline{p_1}^{^{\ell_{s_1}-1}}
        \, \overline{q_2}\, \cdots\, \overline{q_{s_1}}\,
                \overline{p_{s_1}}^{^{\ell_1-1}}\, \overline{q_{s_1+1}}.\]
Observe that such $s_1$-tuples may exist, for example, they do exist if $p_1=q_2=p_2$. 
However, it is easy to decide whether there are such $s_1$-tuples. 
If there are, $\mu$ is not injective on $L_t(K)$. 
We continue with the assumption that such $s_1$-tuples do not exist. 
Then we can decide (i) proceeding as we did above. 
The only difference is that we replace $T_1$ by
\begin{gather*}
T_2=\left\{\left[ \begin{array}{c}
        u_1 \\
        u_2 \\
        \end{array} \right]
\in T_1\colon u_1\neq u_2 \right\}.
\end{gather*}
This is done because we do not want $T_2$ to include words
$\Big[ \begin{array}{c}
        u_1 \\
        u_2 \\
        \end{array} \Big]$
such that
\[u_1= \overline{q_1}\, \overline{p_1}^{^{k_{s_1}-1}}
        \, \overline{q_2}\, \cdots\, \overline{q_{s_1}}\,
                \overline{p_{s_1}}^{^{k_1-1}}\, \overline{q_{s_1+1}},\]
\[u_2=\overline{q_1}\, \overline{p_1}^{^{\ell_{s_1}-1}}
        \, \overline{q_2}\, \cdots\, \overline{q_{s_1}}\,
                \overline{p_{s_1}}^{^{\ell_1-1}}\, \overline{q_{s_1+1}}\]
and
\[(k_1,\ldots,k_{s_1})=(\ell_1,\ldots,\ell_{s_1}).\]
Observe that we did not have this problem in case (ii) because there the
languages $L_t(K_1)$ and $L_t(K_2)$ were disjoint.

\section{Proof of Theorem~\ref{the:main2}}

Let us fix some notation first. If $A_1, A_2, \ldots, A_s$ are matrices, then their {\em direct sum} 
$A_1\oplus A_2 \oplus \cdots \oplus A_s$ is
\[	\left( \begin{array}{cccc}
	A_1 & 0 & \cdots &  0 \\
	0 & A_2 & \cdots &  0 \\
	\vdots & \vdots & \ddots & \vdots \\
	0 & 0 & \cdots & A_s \\
	\end{array} \right).
\]
If $A=(a_{ij})_{m\times n}$ and $B$ are matrices, then their {\em Kronecker product} $A\otimes B$ is
\[	\left( \begin{array}{cccc}
	a_{11}B & a_{12}B & \cdots &  a_{1m}B \\
	a_{21}B & a_{22}B & \cdots &  a_{2m}B \\
	\vdots & \vdots & & \vdots \\
	a_{m1}B & a_{m2}B & \cdots & a_{mn}B \\
	\end{array} \right).
\]
In both cases, we have used block notation.

The direct sum and the Kronecker product have the following properties: 
if $A_1, A_2,$ $\ldots,A_s$ are $m\times m$ matrices and
$B_1, B_2, \ldots, B_s$ are $n\times n$ matrices, then
\[(A_1\oplus B_1)(A_2\oplus B_2)\cdots (A_s\oplus B_s)= (A_1A_2\cdots A_s)\oplus (B_1B_2\cdots B_s)\]
and
\[(A_1\otimes B_1)(A_2\otimes B_2)\cdots (A_s\otimes B_s)= (A_1A_2\cdots A_s)\otimes (B_1B_2\cdots B_s).\]
For more details on the Kronecker product, see for example \cite[Chapter 12]{Lancaster-Tismenetsky} or \cite{Kuich-Salomaa}.

If $k$ is a positive integer, then $E_k=(e_{ij})_{k\times k}$ is the $k\times
k$ matrix whose only nonzero entry is $e_{1k}=1$.

The main idea of our proof of Theorem~\ref{the:main2}
is to use the undecidability of Hilbert's tenth problem combined with the following result. 
Suppose that $t$ is a positive integer and that $p(x_1,\ldots,x_t)$ is a
polynomial with integer coefficients.
We want to find a positive integer $k$ and matrices $A,M,N,B\in\matkZ$ such that
\[AM^{^{a_1}}NM^{^{a_2}}N\cdots NM^{^{a_t}}B=p(a_1,\ldots,a_t)E_k\]
for all nonnegative integers $a_1,\ldots,a_t$.

Fix the value of $t$. 

\begin{lemma}
Let $i\in\{1,\ldots, t\}$. Then there exists a positive integer $k$ and
matrices $A,M,N,B\in\matkZ$ such that
\[AM^{^{a_1}}NM^{^{a_2}}N\cdots NM^{^{a_t}}B=a_iE_k\]
for all nonnegative integers $a_1,\ldots,a_t$.
\end{lemma}

\begin{proof}
Let $k=2t$, 
\[A=\left( \begin{array}{cccc}
	1 & 0& \cdots & 0 \\
	0 & 0&  \cdots & 0 \\
	\vdots & \vdots & & \vdots \\
	0 & 0 & \cdots & 0 \\
	\end{array} \right)
\text{ and }
B=\left( \begin{array}{cccc}
	0 & \cdots & 0 & 0 \\
	0 &  \cdots & 0 & 0 \\
	\vdots & & \vdots & \vdots \\
	0 & \cdots & 0 & 1 \\
	\end{array} \right),
\]
where $A,B\in\matkZ$.
Let $E=\Big( \begin{array}{ll}
	1 & 1 \\
	0 & 1 \\
	\end{array} \Big)$ 
and $I=\Big( \begin{array}{ll}
	1 & 0 \\
	0 & 1 \\
	\end{array} \Big)$ . 
Let
\[M=I\oplus \cdots \oplus I \oplus E \oplus I \oplus \cdots \oplus I,\]
where there are $t$ summands of which $E$ is the $i$th one, and let
\[N=\left( \begin{array}{ccccc}
	0 & I & 0 & \cdots & 0 \\
	0 & 0 & I & \cdots & 0 \\
	\vdots & \vdots & \vdots & & \vdots \\
	0 & 0 & 0 & \cdots & I \\
	0 & 0 & 0 & \cdots & 0 \\
        \end{array} \right)
\]
be a $k\times k$ matrix where each $0$ stands for the $2\times 2$ zero matrix.

Then $A,M,N,B\in\matkZ$ and we have
\begin{eqnarray*}
M^n 
	&=& I\oplus \cdots \oplus I \oplus E^n \oplus I \oplus \cdots \oplus I\\
	&=& I\oplus \cdots \oplus I \oplus \Big( \begin{array}{ll}
											1 & n \\
											0 & 1 \\
										\end{array} \Big) 
				\oplus I \oplus \cdots \oplus I
\end{eqnarray*}
 for all $n\in\N$. 

Now, if $D$ is any matrix in $\matkZ$ then the only nonzero entry of $ADB$ is the last entry in the first row,
which is equal to $D_{1k}$.
Let us compute this entry for \[AM^{^{a_1}}NM^{^{a_2}}N\cdots NM^{^{a_t}}B\]
where $a_1,\ldots,a_t$ are nonnegative integers.
For this, we regard $M$ and $N$ as $t\times t$ matrices consisting of $2\times 2$ blocks:
\begin{align*}
(M^{^{a_1}}N&M^{^{a_2}}N\cdots NM^{^{a_t}})_{1t}\\
	&= (M^{^{a_1}})_{11}N_{12}(M^{^{a_2}})_{22} N_{23}\cdots N_{i-1,i}(M^{^{a_i}})_{ii} N_{i,i+1}\cdots N_{t-1,t}(M^{^{a_t}})_{tt}\\
	&= I\cdot I \cdot I \cdots I \cdot \Big( \begin{array}{ll}
										1 & a_i \\
										0 & 1 \\
									\end{array} \Big) 
					\cdot I  \cdots I\\
	&= \Big( \begin{array}{ll}
				1 & a_i \\
				0 & 1 \\
			\end{array} \Big).
\end{align*}
The results follows.
\end{proof}

\begin{lemma}
Let $p_1(x_1,\ldots,x_t)$ and $p_2(x_1,\ldots,x_t)$ be polynomials with
integer coefficients.
Suppose there exist $s_1, s_2\ge1$, $A_1,M_1,N_1,B_1\in\Z^{s_1\times s_1}_{\rm uptr}$
and $A_2,M_2,N_2,B_2\in\Z^{s_2\times s_2}_{\rm uptr}$ such that
\[A_1M_1^{^{a_1}}N_1M_1^{^{a_2}}N_1\cdots N_1M_1^{^{a_t}}B_1
=p_1(a_1,\ldots,a_t)E_{s_1}\]
and
\[A_2M_2^{^{a_1}}N_2M_2^{^{a_2}}N_2\cdots N_2M_2^{^{a_t}}B_2
=p_2(a_1,\ldots,a_t)E_{s_2}\]
for all $a_1,\ldots,a_t\in\N$.
Then
\begin{itemize}
\item[(i)] there exist $s_3\ge1$ and $A_3,M_3,N_3,B_3\in\Z^{s_3\times s_3}_{\rm uptr}$ such that
\[A_3M_3^{^{a_1}}N_3M_3^{^{a_2}}N_3\cdots N_3M_3^{^{a_t}}B_3
=(p_1+p_2)(a_1,\ldots,a_t)E_{s_3}\]
for all $a_1,\ldots,a_t\in\N$;

\item[(ii)] there exist $s_4\ge1$ and $A_4,M_4,N_4,B_4\in\Z^{s_4\times s_4}_{\rm uptr}$ such that
\[A_4M_4^{^{a_1}}N_4M_4^{^{a_2}}N_4\cdots  N_4M_4^{^{a_t}}B_4 =(p_1\cdot
p_2)(a_1,\ldots,a_t)E_{s_4}\]
for all $a_1,\ldots,a_t\in\N$;

\item[(iii)] if $c\in\Z$, then there exists $A_5\in\Z^{s_1\times s_1}_{\rm uptr}$ such that
\[A_5M_1^{^{a_1}}N_1M_1^{^{a_2}}N_1\cdots N_1M_1^{^{a_t}}B_1 =c\cdot
p_1(a_1,\ldots,a_t)E_{s_1}\]
for all $a_1,\ldots,a_t\in\N$.
\end{itemize}
\end{lemma}

\begin{proof}
To prove (i) we take $M_3=M_1\oplus M_2$, $N_3=N_1\oplus N_2$,
\[A_3=\left( \begin{array}{cccc}
        1 & 1 & \cdots & 1 \\
        0 & 0&  \cdots & 0 \\
        \vdots &  \vdots & & \vdots \\
        0 & 0 & \cdots  & 0 \\
        \end{array} \right)
\cdot (A_1\oplus A_2)\]
and
\[B_3= (B_1\oplus B_2)\cdot \left( \begin{array}{cccc}
        0 & \cdots & 0 & 1 \\
        0 & \cdots & 0 & 1 \\
        \vdots &  & \vdots & \vdots \\
        0 & \cdots  & 0 & 1 \\
        \end{array} \right).\]
To prove (ii) we take $A_4=A_1\otimes A_2$, $M_4=M_1\otimes M_2$,
$N_4=N_1\otimes N_2$ and $B_4=B_1\otimes B_2$.
To prove (iii) it suffices to take $A_5=cA_1$.
Then the claims follow by simple computations which are left to the reader.
\end{proof}

Now our goal is achieved and we can state the following lemma.

\begin{lemma}\label{lem:goal}
Let $t$ be any positive integer and $p(x_1,\ldots,x_t)$ be any polynomial with
integer coefficients.
Then there effectively exists a positive integer $k$ and matrices
$A,M,N,B\in\matkZ$ such that
\[AM^{^{a_1}}NM^{^{a_2}}N\cdots NM^{^{a_t}}B
=\left( \begin{array}{cccc}
	0 & \cdots & 0 & p(a_1,\ldots,a_t) \\
	0 & \cdots & 0 &  0 \\
	\vdots &  &\vdots & \vdots \\
	0 & \cdots & 0 & 0 \\
	\end{array} \right)\]
for all $a_1,\ldots,a_t\in\N$.
\end{lemma}

\begin{remark}
Lemma~\ref{lem:goal} is closely related to the well-known fact stating that if $p(x_1,\ldots,x_t)$ is a polynomial
having integer coefficients, then the series
\[\sum_{n_1,\ldots,n_t\ge0} p(n_1,\ldots,n_t) \ x^{n_1}yx^{n_2}y\cdots yx^{n_t}\]
is $\Z$-rational; see for example \cite{Salomaa-Soittola}. The purpose of Lemma~\ref{lem:goal}  
is to show explicitly that we can get this result using only upper-triangular matrices.
\end{remark}

We will use a strong version of the undecidability of Hilbert's tenth problem as
stated in the following theorem (see Theorem 3.20 in \cite{Rozenberg-Salomaa}.

\begin{theorem}\label{the:Hilbert10}
There is a polynomial $P(x_1,x_2,\ldots,x_m)$ with integer coefficients such that no algorithm exists for deciding 
whether an arbitrary equation of the form
\[P(a,x_2,\ldots,x_m)=0,\]
where $a$ is a positive integer, has nonnegative integers $x_2,\ldots,x_m$ as a solution.
\end{theorem}

For $k=2,3,\ldots$, define the Cantor's polynomials $C_2,C_3,\ldots$ as follows:
\begin{eqnarray*}
C_2(x_1,x_2)&=&\frac 12 (x_1+x_2)(x_1+x_2+1)+x_2,\\
C_{k+1}(x_1,\ldots,x_{k+1})&=&C_2(C_k(x_1,\ldots,x_k),x_{k+1}).
\end{eqnarray*}
These polynomials are injective on $\N^k$. In other words,
for all nonnegative integers $n_1,\ldots,n_k,m_1\ldots, m_k$, 
if $C_k(n_1,\ldots,n_k)=C_k(m_1,\ldots,m_k)$ then $n_1=m_1,\ldots,n_k=m_k$.
Note that the $C_k$'s are not injective on $\Z^k$.

Let $P(x_1,\ldots,x_m)$ be as in Theorem~\ref{the:Hilbert10}. Take a new indeterminate $x_{m+1}$ and define the polynomial 
$Q(x_1,\ldots,x_m,x_{m+1})$ by
\[Q(x_1,\ldots,x_m,x_{m+1})=e\cdot
C_{m+1}(x_1,\ldots,x_m,P(x_1,\ldots,x_m)^2\cdot x_{m+1}),\]
where $e$ is a positive integer chosen such that $Q$ has integer coefficients.

\begin{lemma}\label{lem:PQ}
Let $a$ be a positive integer. Then the equation $P(a,x_2,\ldots,x_m)=0$ 
has a solution in nonnegative integers if and only if 
there exist nonnegative integers $b_2,\ldots,b_{m+1},c_2,\ldots,c_{m+1}$ such that
\begin{equation}\label{lem:eq1}
Q(a,b_2,\ldots,b_{m+1})=Q(a,c_2,\ldots,c_{m+1})
\end{equation}
and
\begin{equation}\label{lem:eq2}
(b_2,\ldots,b_{m+1})\neq (c_2,\ldots,c_{m+1}).
\end{equation}
\end{lemma}

\begin{proof}
Suppose first that there exist $d_2,\ldots,d_m\in\N$ such that 
\[P(a,d_2,\ldots,d_m)=0.\]
Then we have
\[Q(a,d_2,\ldots,d_m,x) =e\cdot C_{m+1}(a,d_2,\ldots,d_m,0)\]
for any $x\in\N$. Hence, if we choose
\[(b_2,\ldots,b_{m+1})=(d_2,\ldots,d_m,1)\ \text{ and }\  (c_2,\ldots,c_{m+1})=(d_2,\ldots,d_m,2),\]
then \eqref{lem:eq1} and \eqref{lem:eq2} hold.

Suppose then that $P(a,d_2,\ldots,d_m)\neq 0$ for all $d_2,\ldots,d_m\in\N$. Suppose that 
\[Q(a,b_2,\ldots,b_{m+1})=Q(a,c_2,\ldots,c_{m+1})\]
where $b_2,\ldots,b_{m+1},c_2,\ldots,c_{m+1}\in\N$. Hence
\[C_{m+1}(a,b_2,\ldots,b_m,P(a,b_2,\ldots,b_m)^2b_{m+1})=C_{m+1}(a,c_2,\ldots,c_m,P(a,c_2,\ldots,c_m)^2 c_{m+1}).\]
Because $C_{m+1}$ is injective on $\N^{m+1}$ we obtain
\begin{equation}\label{lem:eq3}
b_2=c_2,\ldots,b_m=c_m
\end{equation}
and
\[P(a,b_2,\ldots,b_m)^2b_{m+1}=P(a,c_2,\ldots,c_m)^2c_{m+1}.\]
Using \eqref{lem:eq3} and the assumption 
\[P(a,b_2,\ldots,b_m)= P(a,c_2,\ldots,c_m)\neq 0,\]
we obtain $b_{m+1}=c_{m+1}$. Consequently, if $P(a,x_2,\ldots,x_m)=0$ does not
have a solution in nonnegative integers,
then there does not exist $b_2,\ldots,b_{m+1},c_2,\ldots,c_{m+1}\in\N$ such that \eqref{lem:eq1} and \eqref{lem:eq2} hold.
\end{proof}

We are now ready for the proof of Theorem~\ref{the:main2}.

Let $P(x_1,\ldots,x_m)$ and $Q(x_1,\ldots,x_{m+1})$ be as above.
By Lemma~\ref{lem:goal} there is a positive integer $k$ and a morphism
$\mu\colon\Delta^*\to\matkZ$ such that
\[\mu(z_1x^{a_1}yx^{a_2}y\cdots yx^{a_{m+1}}z_2)=Q(a_1,\ldots,a_{m+1})E_k\]
for all $a_1,\ldots,a_{m+1}\in\N$.
For each $a\in\N$ define the morphism $\mu_a\colon\Delta^*\to\matkZ$ by
\[\mu_a(z_1)=\mu(z_1x^ay),\ \mu_a(x)=\mu(x),\ \mu_a(y)=\mu(y) \text{ and }\mu_a(z_2)=\mu(z_2).\]
Then
\[\mu_a(z_1x^{a_2}y\cdots yx^{a_{m+1}}z_2)=Q(a,a_2,\ldots,a_{m+1})E_k\]
for any $a\ge 1$ and $a_2,\ldots,a_{m+1}\in \N$.
By Lemma~\ref{lem:PQ}, for any $a\ge 1$, the morphism $\mu_a$ is injective on
$K_m$ if and only if the equation $P(a,x_2\ldots,x_m)=0$
does not have a solution in nonnegative integers. Now Theorem~\ref{the:main2} follows by
Theorem~{\ref{the:Hilbert10}}.

\section{Concluding remarks}

In the proof of our undecidability result we used singular matrices. On the
other hand, in Theorem~\ref{the:main1} we require that $\mu(z_i)$ is
nonsingular for $i=1,\ldots,t+1$. This assumption plays an essential role in
our proof of the theorem. At present we do not know how to avoid using this
assumption.

The following examples illustrate the situations where some of the matrices
$\mu(z_i)$, $1\le i\le t+1$, are singular. The first two examples show that
the singularity of some $\mu(z_i)$ often implies that $\mu$ is not injective while
the third example shows that this is not always the case. In these examples we
use the notations of Section 3.

\begin{example}
Let $t\ge 2$ and assume that there is an integer $i$, $1\le i \le t-1$, such
that $N_i$ is of the form
$\left( \begin{array}{cc}
        0 & B\\
        0 & C
        \end{array} \right)$,
where $B,C\in \Q$. Then
\[N_iMN_{i+1}=N_iN_{i+1}M,\]
which implies that $\mu$ is not injective on $L_t$.
\end{example}

\begin{example}
Let $t\ge 2$ and assume that there is an integer $i$, $3\le i \le t+1$, such
that $N_i$ is of the form
$\left( \begin{array}{cc}
        A & B\\
        0 & 0
        \end{array} \right)$,
where $A,B\in \Q$. Then
\[MN_{i-1}N_i=N_{i-1}MN_i,\]
which implies that $\mu$ is not injective on $L_t$.
\end{example}

\begin{example}
Let $t\ge 1$ and let
\[N_1=N_2=\cdots =N_t=\left( \begin{array}{cc}
        3 & 1\\
        0 & 1
        \end{array} \right),\,
N_{t+1}=\left( \begin{array}{cc}
        0 & 1\\
        0 & 1
        \end{array} \right),\,
M=\left( \begin{array}{cc}
        3 & 0\\
        0 & 1
        \end{array} \right) .\]
Then for any $m_1,\ldots,m_t\ge 0$ we have
\[N_1M^{^{m_1}}N_2M^{^{m_2}}N_3 \ldots N_tM^{^{m_t}}N_{t+1}=\left( \begin{array}{cc}
        0 & E\\
        0 & 1
        \end{array} \right) \]
where
\[E=3^{m_1+\cdots +m_{t}+t}+3^{m_1+\cdots +m_{t-1}+t-1}+\cdots+3^{m_1+m_2+2}+3^{m_1+1}+1.\]
This implies that $\mu$ is injective on $L_t$.
\end{example}

\section{Acknowledgement}
The major part of this work was achieved when the first author was a member 
of the FiDiPro group of the  FUNDIM research center at the University of Turku. 
The authors warmly thank Juhani Karhum\"aki and Luca Zamboni for their support.

\bibliography{bibliographie}
\bibliographystyle{alpha}
\end{document}